\documentclass{article}
\usepackage[utf8]{inputenc}
\usepackage{amsmath}
\usepackage{amssymb}
\usepackage{amsthm}
\usepackage{comment}
\usepackage{graphicx}
\usepackage[dvipsnames]{xcolor}
\usepackage{authblk}

\newtheorem{proposition}{Proposition}

\newtheorem{example}{Example}
\newtheorem{remark}{Remark}

\usepackage{lineno}
\usepackage{lipsum}

\newtheorem{observation}{Observation}

\usepackage{geometry}
\title{Transaction Ordering Auctions}
\author{Christoph Schlegel}
\affil{Flashbots}
\date{}
\begin{document}
\maketitle
\begin{abstract}
 We study equilibrium investment into bidding and latency reduction for different sequencing policies. For a batch auction design, we observe that bidders shade bids according to the likelihood that competing bidders land in the current batch. Moreover, in equilibrium, in the ex-ante investment stage before the auction, bidders invest into latency until they make zero profit in expectation.  
 We compare the batch auction design to continuous time bidding policies (time boost) and observe that (depending on the choice of parameters) they obtain similar revenue and welfare guarantees.
\end{abstract}
\section{Introduction}
Various transaction ordering policies for roll-up sequencers, exchanges or L1 blockchains have been proposed. The following categories seem to cover the available options well:

\begin{itemize}
\item \emph{First Come First Serve} orders transactions by time stamp of arrival at whatever server orders the transactions. Questions of decentralized implementation aside, this policy seems appealing and intuitive to many. It is the go-to policy in traditional finance and hence users are used to interacting with it. FCFS appeals to basic intuitions about fairness.~\cite{kursawe2020wendy,themis} But there is also an efficiency argument to be made: FCFS provides an incentive to incorporate new external information quickly into the state of the system. 
\item \emph{Bidding Based Ordering}: Orders are processed in batches or blocks.  Transactions within a batch interval are ordered according to a function of the attached bids. The function can (hypothetically) be arbitrarily complex: we could order transactions by bid, we could auction individual slots in the batch, or even allow combinatorial bidding where users express preferences over the entire content of the block.  Also the “bid” can sometimes be interpreted broadly, for example in Ethereum block building, the bid could contain the amount of MEV the user allows the block builder to extract from them.
\item \emph{Random ordering or other "shuffling" protocols}: Randomness~\cite{random}, if implementable, is a means to achieve ex-ante fairness when ranking transactions. In a different direction, verifiable sequencing rules \cite{credible} are designed to make it detectable if a sequencer deviates from the rule. 
\item \emph{Hybrid policies}: The above policies can be mixed and matched. For example, FCFS can be implemented with discrete buckets where orders within a bucket are ordered randomly.  A recent proposal \cite{timeboost}, orders transactions by a scoring rule called “time boost” that scores transactions by a combination of time stamps and bids.
\end{itemize}

While these policies look very different from each other, a substantial aspect of all of them is that they organize a contest for earlier transaction execution among those users that care about it (CEX-DEX arb traders, liquidators, etc.). The term contest here has the usual meaning from economics: users exert effort (investing in latency, spending money on bidding, spamming your server with transactions) to produce a signal (a timestamp, a bid, a set of transactions IDs) and based on these signals, we decide which transactions to include and in which order (and therefore decide who wins the different contests). 

The framing as a contest is helpful, in so far as it gives us an indication of what it means to choose a transaction ordering policy: we organize a contest among users and users maximize whatever signal maximizes their chances of winning the contest. Thus, we need to decide on what dimension we want them to compete: latency investment, expenditure on bidding, the number of transaction requests you receive, increasing entropy, a mixture of all of them, etc? 

In the following, I want to focus on the first two categories of transaction ordering policies, time stamp and bidding based policies, and hybrid policies mixing between the first two categories. This is because these policies have an implicit or explicit focus on efficiency (broadly understood) which is desirable. A maybe under-appreciated aspect of bidding-based policies is that they do not eliminate latency competition completely. Even in a pure batch auction, there still is an advantage for low latency bidders when approaching the end of the batch. This has for example been documented \cite{special} in Ethereum block building, where searchers specialized in CEX-DEX arbitrage need low latency to be competitive in “top of block” MEV.

Following this observation, I analyze the equilibrium bidding and equilibrium latency investment of bidders interacting with a batch auction policy and derive the equilibrium welfare and revenue achieved. In comparison to an analysis that does not take latency into consideration, we observe further bid shading in equilibrium. Bidders lower their bids as a function of the the odds of being included in the current batch. As a second step, I analyze the (ex-ante) latency investment decisions of the bidders interacting with a batch auction. It turns out that in equilibrium bidders will engage in a race to the bottom where all expected profits from participating in the batch auction is compete away through latency investment so that bidders break even in expectation. The analysis is carried out for a first price winner-pay auction. However, the analysis also carries over to other pricing rules such as an all-pay auction which frequently occurs in practice. The analysis of the batch auction is applicable to several different situations relevant in practice,such as the competition for the top of block slot in L1 blockchains or to  L2 sequencers using a batch auction design.

Besides batch auctions, alternative "continuous" bidding policies \cite{timeboost} have been suggested. In these policies, transactions are ordered by time stamp, but bidders can additionally improve their position by bidding. As a second contribution, I also analyze the equilibrium bidding behavior in this kind of market design and compare efficiency and revenue to the batch auction case.
The comparison depends on the choice of parameters, the length of the batch in the batch auction versus the maximal time advantage that can be obtained by bidding. However, choosing the parameters comparably, the two design perform very similarly.

Some disclaimers: The following analysis is purely economic. I abstract away from questions of consensus and implementation and assume that the policies considered can be implemented, because they are run by a trusted centralized sequencer or because we know how to decentralize these policies in a satisfying way. I abstract away from incentive compatibility problems (MEV extraction, censoring etc.) on the side of the party that implements the transaction ordering policy and assume that the policies are implemented as stated. 

\section{Equilibrium Analysis of Bidding and Latency Investment}

The starting point of my analysis is a bidding and/or latency race between two bidders, who each want their transaction to be executed before the other bidder’s transaction (I would expect similar results to hold for more than two bidders). A typical situation that triggers such race could, for example, be an arbitrage opportunity arising through a price discrepancy between an off-chain CEX and an on-chain DEX. Another typical example would be a competition for executing a liquidation. While there are other MEV games played in reality, where bidders have more complicated preferences over transactions orderings than just about how two transactions are ordered relative to each other, these atomic contest for earlier inclusion constitute a large fraction of trading activity on most platforms.  Moreover, many other strategies contain an element of it, as it might be a necessary part of the execution of a more complicated trade.

I assume that the value of earlier execution can be different for the two bidders, for example they could have different amounts of liquidity deployed on different platforms so that they can extract different amounts of value from an arb. However, both of the two bidders should have a non-negative value for having their transaction be executed first. The situation is a race between the two bidders in the sense that the bidder, whose transaction is included later, cannot extract any value. 

I assume that there are two sources of uncertainty for the two bidders:
\begin{enumerate}
\item The bidders are uncertain about the competitor’s value of winning.
\item The bidders are uncertain about each others’ latency. Timestamps are uncertain and random but correlated with (ex-ante) latency investment decisions.
\end{enumerate}
\subsection{Batch auctions}
In a batch auction, all transactions arriving within a pre-specified time window (according to some time stamping scheme) are ordered according to their attached bid, with the highest bid transaction being executed first (if feasible), the second highest bid transaction second (if feasible) etc. A bid cut-off (or reserve price) can be used to bound the total number of transactions being included. A variety of payment rules can be used. In the following, I assume that the competition between the two bidders follow a first price auction format with two possible interpretations: 1) the payment rule is pay-as-you bid, but the capacity is bounded so that the lower bid transaction is not executed 2) the lower bid transaction is reverted if it does not land in the higher slot. Qualitatively very similar results would also hold for an all-pay batch auction, see Remark~2.

 For the analysis I normalize time so that bidding happens during a unit time interval $[0,1]$.  The timing is as follows:

\begin{enumerate}
\item An arbitrage opportunity arrives uniformly at random in the unit interval. When the arbitrage opportunity occurs, bidders learn their valuations. A player $i$ has a valuation $v_i$ to have his transaction included first, where $v_i$ is distributed according to $F_i(x)$.
\item Bidders send a bid for inclusion. Depending on their latency and time of observing the arb, their bid gets included in the current or in the next batch. 
\item At the end of the batch, bids are evaluated according to a first price auction. The higher bid transaction is executed and pays the attached bid.
\end{enumerate}

Assume that there are two bidders with valuations $v_1$ and $v_2$, distributed, for simplicity, i.i.d. uniform on the unit interval. Qualitatively, the analysis carries over to non-uniform valuations. The assumption of independent valuations models the case where there is heterogeneity between the two bidders, while the common component in valuations is known with certainty.

Assume that bidder $i$ has probability of $T_i(\tau)$ of submitting a bid at time $0\leq \tau\leq 1$ which is included in the current batch and probability $1-T_i(\tau)$ of being included in a later batch.

\begin{figure}[t]
\centering
\includegraphics[scale=0.5]{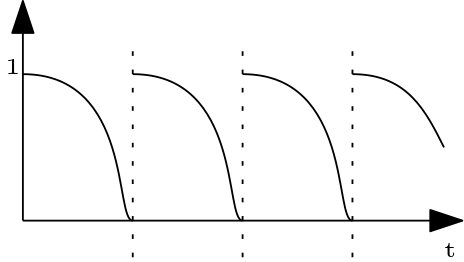}
\caption{Probability  $T_i$ of inclusion in the current batch. Dotted lines indicate batch cut-offs.}
\end{figure}
We can interpret $T_i(\tau)$ either as objective probabilities or as bidders’ (subjective) beliefs. I don't make any particular assumption on $T_i$ other than that $1-T_i(\tau)$ is a differentiable CDF, so that $T_i(0)=1$  and $T_i(1)=0$. It is straightforward to calculate equilibrium bidding strategies, both for 1) the situation where each bidder is certain about whether or not their own bid gets included in the batch but uncertain about the competitor’s bid, and for 2) the situation where bidders are uncertain about their own bid and the competitor’s bid. 
\begin{proposition}
The symmetric equilibrium bidding strategy in the batch auction is  
    
    $$
    b(v)=\frac{v^2}{2(v+\tfrac{1-T(\tau)}{T(\tau)})},
    $$
     which gives a payoff of   $$
    \Pi(\tau)=\tfrac{1}{2}-\tfrac{1}{3}T(\tau).
    $$
\end{proposition}
\begin{proof}
    
    For the case that for bidder $i$ the arb realizes at $\tau$, he knows that he can make into the batch and his valuation is $v_i$, his optimal bid solves
    
    $$
    \max_{0\leq b_i\leq 1}(v_i-b_i)(T_j(\tau)Pr_{v_j\sim F_j}[b_i\geq b_j(v_j)]+(1-T_j(\tau))),
    $$
    
    i.e. bidder $i$ wins if the other bidder lands in the same batch but $i$ bids higher (which happens with probability $(T_j(\tau)Pr_{v_j\sim F_j}[b_i\geq b_j(v_j)]$) or the other bidder lands in a later batch (which happens with probability $1-T_j(\tau)$). Taking first order conditions and assuming that we are in a symmetric equilibrium so that $Pr[b_i\geq b_j]=Pr[v_i\geq v_j]=F(v_i)$, we get:
    
    $$
    \frac{(v_i-b_i(v_i))f(v_i)T_j(\tau)}{b'_i(v_i)}-(T_j(\tau)F(v_i)+(1-T_j(\tau)))=0,
    $$
    
    where $b_i(v)$ is the optimal bid as a function of the valuation and $b_i'(v)$ its derivative. Solving the differential equation obtained from the first order condition condition with the boundary condition $b_i(0)=0$ (if the bidder has no value for the arb he doesn't bid), we get the equilibrium bidding strategy: 
    
    $$
    b(v)=\frac{v^2}{2(v+\tfrac{1-T(\tau)}{T(\tau)})}.
    $$
    
    To calculate the equilibrium pay-off we integrate over the value distribution:
    
    $$
    \Pi(\tau)=\int_{0}^1(v_i-b_i)(T(\tau)F_j(v_i)+(1-T(\tau)))dv_i\\=\int_{0}^1(v_i-\frac{v^2}{2(v+\tfrac{1-T(\tau)}{T(\tau)})})(T(\tau)v_i+(1-T(\tau)))dv_i=\tfrac{1}{2}-\tfrac{1}{3}T(\tau).
    $$\qed
    \end{proof}
    \begin{remark}
  Similarly, we can derive the optimal bid, if bidders are uncertain not only about the competitor’s but also their own bid making it into the current batch. In that case, bidder $i$’s bid solves 
    
    $$
    \max_{0\leq b_i\leq 1}(v_i-b_i)((T_i(\tau)T_j(\tau)+(1-T_i(\tau))(1-T_j(\tau)))Pr_{v_j\sim F_j}[b_i\geq b_j(v_j)]+T_i(\tau)(1-T_j(\tau))).
    $$
    
    By the same kind of calculation, the equilibrium bidding strategy for the symmetric case is: 
    
    $$
    b(v)=\frac{v^2}{2(v+\tfrac{T(\tau)(1-T(\tau))}{T(\tau)^2+(1-T(\tau))^2})}.
    $$
\end{remark}
\begin{remark}
Similarly, we can derive the optimal bid for an all-pay version of the batch auction. 
In that case, bidder $i$’s bid solves 
    
    $$
    \max_{0\leq b_i\leq 1}v_i((T_i(\tau)T_j(\tau)+(1-T_i(\tau))(1-T_j(\tau)))Pr_{v_j\sim F_j}[b_i\geq b_j(v_j)]+T_i(\tau)(1-T_j(\tau)))-b_i.
    $$
    
    By the same kind of calculation, the equilibrium bidding strategy for the symmetric case is: 

    $$
    b(v)=\frac{v^2}{2}T(\tau).
    $$
\end{remark}
We can make the following observations about equilibrium bidding:
\begin{remark}In equilibrium,
\begin{enumerate}
\item bids are lower than in a standard first price auction,
\item  bidders shade bids according to the odds $\tfrac{1-T(\tau)}{T(\tau)}$ that the competitor gets included in a later batch, 
\item  the payoffs for bidders are higher than in a standard first price auction,
\item  the revenue for the auctioneer is lower. 
\end{enumerate}
\end{remark}

\subsection{Ex-ante Latency Investment for Batch Auctions}

Next let us look at the latency investment game induced by the batch auction design. Assume that prior to bidding, bidders invest in latency reduction. The bidders do not know the arrival time and the value of the arb when making the investment decision. The investment is therefore at the ex-ante stage. When making their investment decision, the bidders only know the expected ex-ante profit from the auction as analyzed above.
Assume that it costs $C(\Delta)$ for a user to invest into latency so that he can send messages with delay $\Delta$. We assume that

\begin{enumerate}
\item it is impossible to send arbitrarily fast messages, $\lim_{\Delta\rightarrow 0}C(\Delta)=\infty$,
\item it is possible without any additional investment into latency to send a bid at the beginning of the batch that is included in the batch, $C(1)=0$,
\item lower latency is more expensive than higher latency, i.e.  $C(\Delta)$ is strictly decreasing,
\item there is increasing marginal cost of a marginal time advantage, i.e. $C$ is convex.
\end{enumerate}

Suppose that bidder $i$ invests into latency with delay $\Delta_i$ so that he is able to bid on all arbs arriving up to time $1-\Delta_i$. Assuming uniform arrivals of arbs, defining the probability that bidder $j$ has delay at most $\Delta_j$ by $\sigma_j(\Delta_j)$, bidder $i$’s profit of investing into latency to obtain a delay of $\Delta_i$ is (see the analysis in the previous section):

$$
\int_0^{1-\Delta_i}(\tfrac{1}{2}-\tfrac{1}{3}T(\tau))d\tau-C(\Delta_i)=\int_{\Delta_i}^1(\tfrac{1}{2}-\tfrac{1}{3}\sigma_j(\Delta))d\Delta-C(\Delta_i).
$$
\begin{figure}[t]
\centering
\includegraphics[scale=0.55]
{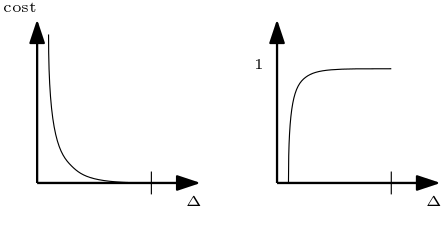}\caption{Cost of latency and the implied equilibrium latency distribution (prob that latency is smaller than $\Delta$.}
\end{figure}
From this we can construct a symmetric equilibrium of the investment game. 
The proofs of all subsequent propositions and observations are in the appendix.
\begin{proposition}
There is a symmetric equilibrium of the ex-ante investment game with mixed strategy
    $$
    \sigma(\Delta)=3/2+3C'(\Delta),
    $$
    in which bidders make $0$ profit on average.
\end{proposition}
\begin{proof}
    
    By symmetry, $\sigma(\Delta):=\sigma_1(\Delta)=\sigma_2(\Delta)$ which describes the (mixed) strategy of each bidder. Since we are in an equilibrium, it has to be the case that the profit is constant on the support of $\sigma$, i.e. for each $\Delta\in [\underline{\Delta},\overline{\Delta}]$ we need
    
    $$
    \frac{\partial}{\partial \Delta}\left(\int_{\Delta}^1(\tfrac{1}{2}-\tfrac{1}{3}\sigma(\Delta))d\Delta-C(\Delta)\right)=0\Rightarrow -\tfrac{1}{2}+\tfrac{1}{3}\sigma(\Delta)-C'(\Delta)=0
    $$
    
    We obtain the mixed strategy:
    
    $$
    \sigma(\Delta)=3/2+3C'(\Delta).
    $$
    
    The support of $\sigma$ is obtained by setting the above formula equal to $0$ resp. $1$ which gives $C'(\underline{\Delta})=-1/2$ and $C'(\overline{\Delta})=-1/6.$
    
    Substituting the strategy into the expression for the profit, we see that bidders make zero average profit
    
    $$
    \int_{\Delta_i}^1(\tfrac{1}{2}-\tfrac{1}{3}\sigma(\Delta))d\Delta-C(\Delta_i)=\int_{\Delta_i}^1-C'(\Delta)d\Delta-C(\Delta_i)=C(\Delta_i)-C(\Delta_i)=0.
    $$\qed
 \end{proof}   
 Ex-ante investment into latency leads to zero average profit for bidders in equilibrium: The ex-ante expected gains from bidding in the batch auction are equal to the cost of latency investment of the bidder independently of the cost function.

Finally, we are interested in the expected value of the latency difference $\Delta:=|\Delta_2-\Delta_1|$ induced by the latency investment. This quantity is relevant for the welfare and revenue guarantees of the batch auction. We have
\begin{align*}E_{\sigma}[|\Delta|]=E_{\sigma}[\max\{\Delta_1,\Delta_2\}]-E_{\sigma}[\min\{\Delta_1,\Delta_2\}]=\int_{\underline{\Delta}}^{\overline{\Delta}}2\Delta \sigma'(\Delta)\sigma(\Delta)d\Delta-\int_{\underline{\Delta}}^{\overline{\Delta}}2\Delta \sigma'(\Delta)(1-\sigma(\Delta))d\Delta\\=\int_{\underline{\Delta}}^{\overline{\Delta}}2\Delta \sigma'(\Delta)(2\sigma(\Delta)-1)d\Delta=\int_{\underline{\Delta}}^{\overline{\Delta}}6\Delta C''(\Delta)(2+6C'(\Delta))d\Delta\end{align*}

\begin{example}
Consider a cost function of the form $C(\Delta)=c/\Delta$. Then $C'(\Delta)=-c/\Delta^2$ and $C''(\Delta)=2c/\Delta^3$. We get the strategy:
$$\sigma(\Delta)=3/2-3c/\Delta^2$$
for $\sqrt{2c}<\Delta<\sqrt{6c}$.
Then the average latency difference is given by the expectation of the absolute value of $\Delta:=\Delta_2-\Delta_1,$
\begin{align*}E[|\Delta|]\approx 0.3203*\sqrt{c}\end{align*}
Accordingly, the average welfare gap to optimal bidding is  $0.16016\sqrt{c}$ and the average revenue is $1/3-0.10677*\sqrt{c}.$

\end{example}

\subsection{Hybrid Policies: Time Boost}

\cite{timeboost} describe a hybrid transaction ordering policy that orders transaction by time stamp of arrival at the (centralized) sequencer, but allows bidders to get a time boost by paying an additional fee. The time boost for paying a fee of $F$ is

$$
\pi=\frac{gF}{F+c},
$$

with parameters $g$ and $c$, and transactions are ordered by the score 

$$
\pi-t
$$

where $t$ is the timestamp of arrival of a bid at the sequencer and $\pi$ is the time boost. The parameters have the following interpretation: Parameter $g$ gives the maximal time boost, a user can get from bidding. In particular, transactions finalize after waiting $g$. Parameter $c$ is the marginal cost per unit of time (normalized by $g$). The auction is all-pay: Bidders need to pay the time boost fee no matter how transactions are ordered.

%[Time boost formula (source: [Arbitrum research](https://research.arbitrum.io/t/time-boost-a-new-transaction-ordering-policy-proposal/8173))](https://images.app.goo.gl/aXcjUYuy3G1uCcX46)

%Time boost formula (source: [Arbitrum research](https://research.arbitrum.io/t/time-boost-a-new-transaction-ordering-policy-proposal/8173))

Let us again consider the scenario where arb opportunities appear uniformly at random on a time interval and assume that the arb is only good for the earlier transaction. If a bidder has a lower score he will get a payout of zero, but still needs to pay.

The analysis of time boost is more complex than the previous analysis for the batch auction, since now the precise time stamps of the two bidders matter and not only whether the time stamp is before the batch cut-off or not. Thus, making the same analysis as previously, where the values and time stamps of the bidder follows some distribution is generally intractable. However, we can analyze the case where the value of the arb to the other bidder is uncertain, but the time stamps are certain.\footnote{The other extreme case, where the value of the arb is commonly known by the bidders, but the time stamps are uncertain, is also tractable to analyse, see \cite{mamageishvili2023shared} for an analysis.} This simplification is not innocent, but it is a reasonable approximation of reality when bidders interact repeatedly and can estimate the expected time stamp of the competitor with very good accuracy.  For simplicity, I consider a linear approximation for the boost formula:

$$
F=\frac{c\pi}{g-\pi}\approx \frac{c\pi}{g}.
$$

For a reasonably large boost parameter, e.g. $g=10$ when $c=1$ the marginal cost $\frac{cg}{(g-\pi)^2}$ (which is relevant for deriving optimal bidding policies) is approximated very well by $c/g$ and we can expect to get qualitatively very similar results for the true boost formula, as long as $g$ is sufficiently large. In the following, I assume $g\geq c$. The equilibrium analysis now follows an all-pay auction with a head start for the lower latency bidder.

\begin{proposition}
The equilibrium signaling strategies in time boost are
    
    $$
    \pi_1=\begin{cases}\tfrac{gv_1^2}{2c}-\tfrac{\Delta}{2},\quad &\text{ if }v_1\geq\sqrt{\tfrac{c}{g}\Delta},\\0&\text{ if }v_1<\sqrt{\tfrac{c}{g}\Delta},\end{cases}
    $$
    
    and
    
    $$
    \pi_2=\begin{cases}\tfrac{gv_2^2}{2c}+\tfrac{\Delta}{2},\quad &\text{ if }v_2\geq\sqrt{\tfrac{c}{g}\Delta},\\0&\text{ if }v_2<\sqrt{\tfrac{c}{g}\Delta},\end{cases}
    $$
\end{proposition}

\begin{proof}
    
    First I solve the model when the bidders have the same latency and produce the same timestamp. Then this is just a standard all pay auction. Each bidder solves
    
    $$
     \max_{\pi}v_iPr_{v_j\sim F_j}[\pi\geq\pi_j(v_j)]-\tfrac{c}{g}\pi,
    $$
    
    where $\pi_2(v_2)$ is bidder $j$’s time boost bid as a function of his valuation.
    
    Assuming differentiability of the equilibrium bidding functions, equilibrium bidding is characterized by the first order condition:
    
    $$
    \frac{v_if_j(v_i)}{\pi'_i(v_i)}-\tfrac{c}{g}=0
    $$
    
    Solving the differential equation obtained from the first order condition condition with the boundary condition $\pi_i(0)=0$ (if the bidders are symmetric, both start bidding at the zero valuation), we get (for uniformly distributed valuations):
    
    $$
    \pi(v)=\frac{gv^2}{2c}
    $$
    
    In the case that that there is a latency difference $\Delta$, bidder $1$ solves
    
    $$
    \max_{\pi}v_1Pr_{v_2\sim F_2}[\pi+\Delta\geq\pi_2(v_2)]-\tfrac{c}{g}\pi=v_1F_2[v_2(\pi+\Delta)]-\tfrac{c}{g}\pi,
    $$
    
    and  $v_j(\pi)$ is the inverse of $\pi_j$ (assuming that we are in a separating equilibrium such that such inverse exists). Similarly bidder $2$ solves
    
    $$
    \max_{\pi}v_2Pr_{v_1\sim F_1}[\pi-\Delta\geq\pi_1(v_1)]-\tfrac{c}{g}\pi=v_2F_1[v_1(\pi-\Delta)]-\tfrac{c}{g}\pi.
    $$
    
    I use the following approach to construct an equilibrium for the asymmetric case. The equilibrium strategies are obtained by subtracting a constant term from the optimal strategy in the symmetric case for the first bidder
    
    $$
    \pi_1(v)=\frac{gv_1^2}{2c}-K_1\Leftrightarrow v_1(\pi)=\sqrt{\frac{2c(\pi+K_1)}{g}}
    $$
    
    and adding a constant term for the second bidder 
    
    $$
    \pi_2(v)=\frac{gv_2^2}{2c}+K_2\Leftrightarrow v_2(\pi)=\sqrt{\frac{2c(\pi-K_2)}{g}}
    $$
    
    The first order condition for the first bidder is
    
    $$
    v_1v_2'(\pi+\Delta)-c/g=v_1\sqrt{\frac{c}{2g(\pi+\Delta-K_2)}}-c/g=0\Rightarrow \Delta=K_1+K_2
    $$
    
    and for the second bidder is
    
    $$
    v_2v_1'(\pi-\Delta)-c/g=v_1\sqrt{\frac{c}{2g(\pi-\Delta+K_1)}}-c/g=0\Rightarrow \Delta=K_1+K_2.
    $$
    
    There should be a smallest threshold $u$ such that bidders only bid if their valuation is above the threshold: $v_i\geq u$. To guarantee continuity of the expected payout at $v_1=u$, the faster bidder should bid $\pi_1=0$ at $v_1=u$ to obtain an expected payout of $u^2$. In that case:
    
    $$
    K_1=gu^2/(2c) \text{ and }K_2=\Delta-gu^2/(2c)
    $$
    
    To guarantee continuity of the expected payout at $v_2=u$, the weak bidder should be indifferent between bidding and not bidding at the threshold:
    
    $$
    u^2/2-c\pi_2(u)/g=0\Rightarrow K_2=gu^2/(2c)=K_1.
    $$
    
    Therefore: $K_1=K_2=\Delta/2$ and $u=\sqrt{c\Delta/g}.$ We get the equilibrium bidding strategies.\qed
 \end{proof}
   \begin{remark}
   The Arbitrum proposal considers an all-pay auction format where bids are not reverted if the other bidder
obtains a higher score. While the choice of the all pay auction format is caused by technical requirements
(the sequencer is agnostic about the content of the transactions and cannot in- or exclude them based on
un-observable characteristics), it is still interesting to check what would happen in the Arbitrum proposal if
the payment rule is a standard first price auction. By a completely analogous calculation, we obtain:
  $$
    \pi_1=\begin{cases}\tfrac{gv_1}{2c}-\tfrac{\Delta}{2},\quad &\text{ if }v_1\geq{\tfrac{c}{g}\Delta},\\0&\text{ if }v_1<{\tfrac{c}{g}\Delta},\end{cases}
    $$
    
    and
    
    $$
    \pi_2=\begin{cases}\tfrac{gv_2}{2c}+\tfrac{\Delta}{2},\quad &\text{ if }v_2\geq{\tfrac{c}{g}\Delta},\\0&\text{ if }v_2<{\tfrac{c}{g}\Delta},\end{cases}
    $$
   \end{remark}

%\begin{remark} In equilibrium:
%\begin{enumerate}
%\item For low valuations neither bidder will make a time boost bid. 
%\item For high valuations, the bidders produce the same score for the same valuation. 
%\item The value threshold at which bidders start bidding is increasing in $c$, decreasing in $g$ and increasing in the latency difference.
%\item A lower latency bidder underbids relative to a standard all-pay auction, a high latency bidder overbids relative to a standard all-pay auction.
%\item  To maximize bidding revenue, the parameter $c$ should be chosen small and the maximal time boost $g$ large.
%\end{enumerate}

%\end{remark}

An immediate implication of the previous calculations is that the parameter $c$ in the time boost formula should be selected small to increase participation and revenue. In a similar way, the maximal time boost $g$ should be selected large. However, the latter comes with trade-offs, as finalization of bids is slower. I comment on these parameter choices and compare the performance in the next section.

\subsection{Comparing batch auctions and time boost}

It is instructive to compare the relative performance of the two auction formats. To make the batch auction comparable to the time boost proposal, I assume that if neither of the two bidders make it into the batch because they have too high latency, then their orders are both processed in the next batch. We can then look into the equilibrium in either model as a function of the realized difference in latency $\Delta:=|\Delta_2-\Delta_1|$ between the two bidders where $\Delta_i$ is the delay of bidder $i$ when sending a bid to the sequencer. 

First, let us look at allocative efficiency: how likely is it that the higher valuation bidder wins the race?  In the batch auction this can only happen if the two bidders end in different batches and the faster bidder has the lower valuation which happens with probability $\Delta/2$. Under time boost, this can only happen,  if the slower bidder refrains from bidding (which in our equilibrium happens if his valuation is below the threshold $u:=\sqrt{\tfrac{c}{g}\Delta}$) while having a higher valuation. This leads to the following observation:

\begin{observation}
Under the batch auction design, the likelihood that the high valuation bidder loses is half the latency difference $\tfrac{\Delta}{2}$.

Under the time boost design, the likelihood is half the cost of compensating for the latency difference $\tfrac{c}{g}\tfrac{\Delta}{2}$. 

The likelihood of the higher valuation bidder winning is higher under time boost if and only if the marginal cost of bidding is small $\tfrac{c}{g}\leq1$.
\end{observation}

Next let us look at the welfare loss relative to the first best where the item is always allocated to the higher valuation bidder: under which design is this welfare gap larger? In the first best, the surplus is the expectation of the maximum of the two valuations which is $2/3$ for the i.i.d. uniform case. 

\begin{observation}
Under the batch auction design, the welfare gap to the optimum is $\Delta/6$.

Under the time boost design, the welfare gap is $\tfrac{1}{6}(\tfrac{c}{g}\Delta)^{3/2}$.

The welfare gap is smaller with time boost, if and only if the marginal cost of bidding is small  $\tfrac{c}{g}\leq1/\sqrt[3]{\Delta}$.
\end{observation}
\begin{proof}
   Assume w.l.o.g. $\Delta_1\leq\Delta_2$ in the following calculations. For the batch auction design we have an expected surplus of:
    
    $$
    \Delta E[v_1]+(1-\Delta)E[\max\{v_1,v_2\}]=\Delta/2+(1-\Delta)\tfrac{2}{3}=2/3-\Delta/6.
    $$
    
    For the time boost auction we have an expected surplus of:
    
\begin{align*}
    Pr[v_2\leq u] E[v_1]+Pr[v_1\leq u\leq v_2]E[v_2|v_2\geq u]+Pr[v_1\geq u,v_2\geq u]E[\max\{v_1,v_2\}|v_1\geq u,v_2\geq u]\\=u/2+u(1-u)\tfrac{1+u}{2}+(1-u)^2\tfrac{u+2}{3}=2/3-\tfrac{u^3}{6}=2/3-\tfrac{1}{6}(\tfrac{c}{g}\Delta)^{3/2}
    \end{align*}
\end{proof}

Next let's look at bidding revenue for the auctioneer. The ex-ante revenue from bidding is the sum of payments received.

\begin{observation}
     Under the batch auction design, the revenue is $(1-\Delta)/3$.

Under the time boost design, the revenue is $(1-(\tfrac{c}{g}\Delta)^{3/2})/3$.

The revenue is higher under time boost if and only if the marginal cost of bidding is small $\tfrac{c}{g}\leq1/\sqrt[3]{\Delta}$ .
\end{observation}
\begin{proof}
    For the batch-auction design, bidders bid as in a standard first price auction if they expect the other bidder to land in the same batch, and $0$ otherwise (or the minimal bid/fee if there is one). A standard first price auction has an expected revenue of $1/3$. Thus the expected revenue is $(1-\Delta)/3$. Under time boost, given the equilibrium derived previously, the expected payment for the faster bidder is
    
    $$
    Pr[v_1\geq u]E[\tfrac{v_1^2-u^2}{2}|v_1\geq u]=1/6 - u^2/2 + u^3/3,
    $$
    
    and that for the slower bidder is
    
    $$
    Pr[v_2\geq u]E[\tfrac{v_2^2+u^2}{2}|v_2\geq u]=1/6 + u^2/2 - 2 u^3/3.
    $$
    
    Therefore the sum of the two is
    
    $$
    (1-u^3)/3=(1-(\tfrac{c}{g}\Delta)^{3/2})/3.
    $$
  \end{proof}

For the previous comparison, note that I have normalized the size of batches to $1$, whereas the $g$ parameter, which plays a similar role as the batch size, is allowed to be larger than $1$. To get an intuition what happens with variable batch size, note that if a unit time interval is subdivided into two batches while there is still one arb per unit of time, then the likelihood that the lower valuation bidder wins the race doubles. Similarly, the welfare gap grows and the revenue decreases in the number of batches. Thus, choosing larger batches has a qualitatively very similar effect as choosing a larger maximal time boost $g$.
\section{Conclusion}
The previous analysis is stylized but has immediate implications for economic design:
While latency competition is much less severe for the batch auction than in a FCFS design, there are still advantages of low latency bidding. Towards the end of a batch, bidders with a faster connection can underbid relative to optimal bidding in a standard first price auction, since there is a substantial likelihood that slower competitors do not make it into the current batch. This has an adverse effect on efficiency and revenue. Moreover, the equilibrium analysis predicts that all profits of the bidders are competed away through latency investment by the bidders in the ex-ante stage before the actual bidding.

If the time boost design is implemented, special attention should be put on the choice of parameters. Choosing the marginal cost of bidding too high is detectable in equilibrium. In that case we can predict little use of time boost bidding (no bidding rather than producing low signals) and the parameters should be adjusted. 

\bibliographystyle{ACM-Reference-Format}
\bibliography{name}

\end{document}